\documentclass[11pt]{amsart}

\usepackage[USenglish]{babel}

\usepackage{amsmath,amsthm,amssymb,amscd}
\usepackage{booktabs}

\usepackage{MnSymbol}

\usepackage{tikz}
\usetikzlibrary{arrows,shapes.geometric,positioning,decorations.markings}

\usepackage[mathscr]{eucal}
\usepackage[normalem]{ulem}
\usepackage{latexsym,youngtab}
\usepackage{multirow}
\usepackage{epsfig}
\usepackage{parskip}

\usepackage[margin=2cm]{geometry}
\usepackage{color}

\newtheoremstyle{custom}
  {3pt}
  {3pt}
  {\slshape}
  {}
  {\bfseries}
  {.}
  { }
   {}
\theoremstyle{custom}
\newtheorem{theorem}{Theorem}[section]
\newtheorem{proposition}[theorem]{Proposition}

\newtheorem{proposition/definition}[theorem]{Proposition/Definition}
\newtheorem{lemma}[theorem]{Lemma}

\theoremstyle{definition}
\newtheorem{definition}[theorem]{Definition}

\theoremstyle{remark}
\newtheorem{remark}[theorem]{Remark}

\newtheoremstyle{exercise}
  {3pt}
  {6pt}
  {}
  {}
  {\bfseries}
  {:}
  { }
   {}
\theoremstyle{exercise}
\newtheorem{exercise}[theorem]{Exercise}
\newtheoremstyle{exercises}
  {3pt}
  {6pt}
  {}
  {}
  {\bfseries}
  {:}
  {\newline}
   {}

\theoremstyle{exercise}
\newtheorem{exercises}[theorem]{Exercises}

\input epsf
\def\boxit#1{\vbox{\hrule height1pt\hbox{\vrule width1pt\kern3pt
  \vbox{\kern3pt#1\kern3pt}\kern3pt\vrule width1pt}\hrule height1pt}}

\def\trank{\text{rank}}

\def\BC{\mathbb C}

\def\tdim{{\rm dim}}

\def\hd{,...,}
\def\ww{\wedge}

\def\inv{{}^{-1}}

\def\cB{{\mathcal B}}\def\cA{{\mathcal A}}

\def\CC{\mathbb C}

\def\11{\mathbf 1}

\def\ot{{\mathord{ \otimes } }}

\def\ra{{\mathord{\;\rightarrow\;}}}

\def\dim{{\rm dim}\;}
\def\La#1{\Lambda^{#1}}

\def\BZ{\Bbb Z}

\def\ep{\epsilon}

\def\BC{\mathbb  C}

\def\ep{\epsilon}

\def\hd{, \hdots ,}

\def\inv{{}^{-1}}

\def\La#1{\Lambda^{#1}}

\def\ur{\underline{\mathbf{R}}}
\def\uR{\underline{\mathbf{R}}}

\def\ra{\rightarrow}

\def\tdim{\operatorname{dim}}

\def\tmax{\operatorname{max}}

\def\trank{\operatorname{rank}}

\def\ww{\wedge}

\def\bbb{{\mathbf{b}}}

\def\be{\begin{equation}}
\def\ene{\end{equation}}
\def\aaa{{\mathbf{a}}}
\def\bbb{{\mathbf{b}}}
\def\ccc{{\mathbf{c}}}

\DeclareMathOperator{\tlog}{log}

\def\trank{\mathbf{R}}

\def\vp{\mathbf{V}\mathbf{P}}

\def\vnp{\mathbf{V}\mathbf{N}\mathbf{P}}

\newcommand{\Id}{\operatorname{Id}}

\def\Mn{M_{\langle \nnn \rangle}}

\def\Mn{M_{\langle \nnn \rangle}}

\def\trank{{\mathrm {rank}}}

\def\aaa{\mathbf{a}}
\def\bbb{\mathbf{b}}
\def\ccc{\mathbf{c}}

\def\nnn{\mathbf{n}}

\def\Mn{M_{\langle \nnn \rangle}}

\def\trank{{\mathrm {rank}}}

\def\aaa{{\bold a}} \def\ccc{{\bold c}}

\def\nnn{\bold n}

\def\tsupp{{\rm supp}}

\keywords{Matrix multiplication complexity, Tensor rank, Asymptotic rank, Laser method}


\renewcommand{\BC}{\mathbb{C}}

\newcommand{\bfR}{\mathbf{R}}

\newcommand{\frakgl}{\mathfrak{gl}}

\newcommand{\textsum}{{\textstyle\sum}}

\begin{document}

\author[J.M. Landsberg and Mateusz Michalek]{J.M. Landsberg and Mateusz Micha{l}ek}

\address{Department of Mathematics, Texas A\&M University, College Station, TX 77843-3368, USA}
\address{Max Planck Institute for Mathematics in the Sciences, Leipzig, 04103, Germany}
 \email[J.M. Landsberg]{jml@math.tamu.edu}
\email[M. Micha{l}ek]{Mateusz.Michalek@mis.mpg.de}

\title[Tensor of high border rank]{Towards finding hay in a haystack: explicit tensors of border rank greater than $2.02m$
in $\BC^m\ot \BC^m\ot \BC^m$}

\thanks{Landsberg   supported by NSF grant  AF-1814254.}



\begin{abstract} 
We write down an explicit sequence of tensors in $\BC^m\ot \BC^m\ot\BC^m$, for all $m$ sufficiently large, having border rank
  at least $2.02m$, overcoming a longstanding barrier. We obtain our lower bounds
via the border substitution method.
 \end{abstract}

\maketitle

\section{Introduction}
A frequently occurring theme in algebraic complexity theory is the {\it hay in a haystack} problem (phrase due
to H. Karloff): find an explicit object that behaves generically. For example Valiant's $\vp$ vs.~$\vnp$ problem
is to find an explicit polynomial sequence that is difficult to compute.
We address this problem for three-way tensors. Here the state of the art is embarrassing.

Let $A,B,C$ be complex vector spaces. A tensor $T\in  A \ot B\ot C$ has {\it rank one} if $T=a\ot b\ot c$ for some 
$a\in A$, $b\in B$, $c\in C$. The {\it rank} of $T$, denoted $\bfR(T)$, is the 
smallest $r$ such that $T$ is a sum of $r$ rank one tensors.
 The {\it border 
rank} of $T$, denoted $\ur(T)$, is the smallest $r$ such that $T$ is the limit of a 
sequence of rank $r$ tensors.

If a tensor  $T\in \BC^m\ot \BC^m\ot \BC^m$ is chosen randomly, with $m>3$, then with probability
one, it will have  rank and border rank $\lceil \frac{m^3}{3m-2}\rceil\sim \frac{m^2}3$ \cite{MR87f:15017}.  More precisely, the
set  of tensors with    border rank less than $\lceil \frac{m^3}{3m-2}\rceil$ is a proper subvariety and the
set of tensors with rank  not equal to  $\lceil \frac{m^3}{3m-2}\rceil$ is contained in a proper subvariety. 
Previous to this paper, there was no explicit sequence of tensors of
 border rank at least $2m$ known (and no explicit
sequence of
rank at least $3m$ known), although several known sequences come close to these.
In \cite{MR3025382} they exhibit an explicit sequence of tensors satisfying $\bold R(T_m)\geq 3m-O(\tlog_2(m))$
and in \cite{MR3320240} a sequence of tensors of border rank at least $2m-2$ is presented, although
these tensors are only \lq\lq mathematician explicit\rq\rq\ and not \lq\lq computer scientist explicit\rq\rq\ 
as some entries of the tensor are of the form $2^{2^m}$. We explain  the notion of explicit in \S\ref{subs:whatEX}.
The best border rank lower bound for 
a computer science explicit tensor to our knowledge is ironically the tensor $\Mn$ corresponding 
to matrix multiplication of $\nnn\times \nnn$
matrices where, setting $m=\nnn^2$,  $\ur(\Mn)\geq 2m-\lceil\frac 12 \tlog_2(m)\rceil-1$ \cite{MR3842382}. In this paper
we deal exclusively with border rank. We show:

\begin{theorem} \label{mainthm}
For each  $k$,    define the $\lceil \frac{3(2k+1)^2}{4}\rceil$-dimensional
family of tensors $T_k=T_k(p_{ij})\in \BC^{2k+1}\ot\BC^{2k+1}\ot\BC^{2k+1}=: A\ot B\ot C$, where $|i|,|j|, |i+j|\leq k$ as follows:
\be\label{Tk}
T_k=\sum_{i=-k}^k \sum_{j=\max(-k,-i-k)}^{\min(k,-i+k)}p_{i j} a_i\ot b_j\ot c_{-i-j},
\ene
where $A\simeq B\simeq C\simeq \BC^{2k+1}$ and $(a_{-k},\dots, a_k),(b_{-k},\dots, b_k),(c_{-k},\dots, c_k)$ 
are respectively the bases of the vector spaces
$A,B,C$.  

Then for all $0<\ep<\frac 1{42}$, and all  $k$ at least
$$
\frac{ 37044\ep^3 - 82908\ep^2 - 983829\ep + 364175}{(1-42\ep)^3},
$$
there exists an explicit assignment of the $p_{ij}$ 
in the sense of Definition \ref{explicitdef}, so that
$\ur(T_k)\geq (2+\ep)m$
In particular, when 
$\ep=.001$ this holds once   $k\geq 413085$, when 
$\ep=.01$,  this holds once   $k\geq 1.82\times 10^6$,  
 and when $\ep=.02$  this holds once  
$k\geq 8.41\times 10^7$.
\end{theorem}

The explicit assignment is as follows: for $i>\frac 37k$, list the $p_{ij}$ in some order and assign them  
distinct  prime numbers as in Lemma \ref{expla}. For $i\leq \frac 37k$, we assign them periodically as in Lemma \ref{explb}. Precisely, for a fixed constant $K$ introduced in \eqref{Keqn}, the value of $p_{ij}$ depends only on $i \text{ mod } 4K$ and $j \text{ mod } 4K$. For $0\leq i,j<4K$ the values of the $p_{ij}$'s are determined in such a way that no subset of distinct $p_{ij}$'s is a solution to any of the explicit nonzero polynomials that we derive in Lemma \ref{neardiagcase}.

If $k$ is smaller, we can still get border rank bounds of twice the dimension or greater. For example:

\begin{theorem}\label{thm:explicit}
Let $T_6$ be assigned the values $p_{ij}=2^i+3^j-7ij$.
Then $\ur(T_6)\geq 26=2m$.
\end{theorem}

\subsection*{Acknowledgements}
We thank M. Bl\"aser,  M. Forbes and J. Grochow for explaining various definitions of
explicitness to us.

\section{Background} We explain why the family $T_k$ is promising for proving border rank lower bounds.

For a tensor $T=\sum_{ijk}t_{ijk} a_i\ot b_j\ot c_k$ given in bases, the {\it support} of $T$
is $\tsupp(T):=\{ (ijk)\mid t_{ijk}\neq 0\}$.
Let $\aaa=\tdim A$, $\bbb=\tdim B$, $\ccc=\tdim C$. 

\begin{definition}\cite{MR1089800}
A tensor $T\in A\ot B\ot C$ is {\it tight}
if its symmetry Lie algebra contains a regular semisimple element of $\frakgl(A) \oplus \frakgl(B) \oplus \frakgl(C)$,
under its natural action on $A \otimes B \otimes C$. 
In other words, $T$ is tight if
there exist bases of $A,B,C$  and  
injective functions $\tau_A : [\aaa] \to \BZ$, $\tau_B : [\bbb] \to \BZ$ and $\tau_C : [\ccc] \to \BZ$ such that
for all $(i,j,k) \in \tsupp(T)$, $\tau_A(i) + \tau_B(j) + \tau_C(k)=0$. (See \cite{CGLVW} for details.)
 Call such a basis a {\it tight basis}.
\end{definition}

The border substitution method and normal form lemma of \cite{MR3633766} imply:

\begin{proposition}\label{tightbapolarok}
Let $T\in A\ot B\ot C$. Then there is a filtration $A_1^*\subset \cdots \subset A_{\aaa-1}^*\subset A^*$
of $A^*$, such that letting $T_q:=T|_{A_{\aaa-q}^*\ot B^*\ot C^*}$,
$$
\ur(T)\geq \tmax_{q\in \{1\hd \aaa-1\}}(q+\ur(T_{q})).
$$
If moreover $T$ is tight, then there exists such a filtration where each filtrand is spanned by tight basis elements.
\end{proposition}

Thus tight tensors are a natural class of tensors where the border substitution method may be implemented.

If $T$ is tight, then any tensor with the same support as $T$ in a tight basis is also tight. In particular, it
is natural to study tight tensors in families indexed by the support. In this context, 
it was shown in \cite{CGLVW} that the largest possible  support  
for a tight tensor is given by the   family \eqref{Tk},  which  is a $\lceil \frac 34 m^2\rceil$-dimensional
family. That is, the family \eqref{Tk} is the largest such that Proposition \ref{tightbapolarok} may be applied to.

\begin{remark}
There is an absolute limit to the utility of the border substitution method combined with existing
lower bound techniques of proving border rank bounds which is  roughly  $3m-3\sqrt{3m}$ \cite{MR3842382}.
\end{remark}

\section{Explicit tensors}\label{subs:whatEX}

\begin{definition} \label{explicitdef}
 A sequence of objects (e.g., numbers, graphs, tensors), indexed by integers, where the $n$-th object has size
  (i.e.,  the minimal number of bits needed to specify it) at least $f(n)$ 
 is {\it explicit}  if the $n$-th object may be computed in time polynomial   in $f(n)$. 
\end{definition}

We remark that there are more restrictive definitions of
what it means for a sequence to be explicit but the experts we asked viewed the above one as acceptable.

Since  tensors in $\CC^m\ot\CC^m\ot\CC^m$ have size at least $m^3$, a 
 sequence of tensors is explicit if the $m$-th tensor in $\CC^m\ot\CC^m\ot\CC^m$ may be computed in time polynomial   in $m$.

\begin{lemma}\label{expla} Let $T_m\in \BC^m\ot \BC^m\ot \BC^m$ be a sequence of tensors whose only nonzero entries
are as in \eqref{Tk}. If the entries of $T_m$ are the first distinct $m^2$ prime numbers, then $T$ is explicit.
\end{lemma}
\begin{proof}
Let $\pi(x)$ be the number of prime numbers smaller or equal than $x$.  By \cite{rosser1962approximate}, 
  $\pi(x)>\frac{x}{\log x}$ for $x\geq 17$. Hence, running the sieve method on first $m^3$ numbers, 
we may find the  first $m^2$ prime numbers in polynomial time. 
\end{proof}

The following is obvious but we record it for future use:

\begin{lemma}\label{explb}Let $T_m\in \BC^m\ot \BC^m\ot \BC^m$ be a sequence of tensors whose only nonzero entries
are as in \eqref{Tk}. If the entries of $T_m$ 
are taken from a fixed library $\{ c_{a,b}\}_{0\leq a< \cA, 0\leq b< \cB}$ of
integers periodically, so $p_{ij}=c_{(i\text{ mod }\cA),(j\text{ mod }\cB)}$), then $T_m$ is explicit.
\end{lemma}

The tensors in our  sequence will be a sum of tensors of the types from  Lemmas \ref{expla} and \ref{explb}, and thus
also explicit.

\section{Outline of the proof}
We need to 
  find explicit $p_{ij}$ for the sequence \eqref{Tk}  and to lower bound the border ranks of   $T_{k,q}$
  as in Proposition \ref{tightbapolarok} for all
sufficiently large $k$ and some
useful $q=q(k)$.
We first determine $q$.

Let $I_q$ denote the index set kept  to obtain $T_{k,q}$,  so $|I_q|=2k+1-q$.
Write  $I_q=I_{-}\sqcup I_0\sqcup I_{+}$ for the negative, zero, and positive indices appearing in $I_q$.
Fix a large constant $K$   to be determined later.
In what follows $0<\ep<\frac 1{42}$.

{\bf Set up}: Let $q$ be the largest index   such that in each of $I_-,I_+$ there exist
one of 
\begin{enumerate}
\item Three indices with absolute value at least $\frac k2$, or
\item Five indices with absolute value smaller than $\frac k2$ such that the
absolute value of the difference
of the smallest and largest one is at most $K$.
\end{enumerate}

Without loss of generality we may assume that for $q+1$ both conditions fail for $I_-$. In that 
case we have $|I_-|\leq \frac{2k}{K}+7$ because there can be at most 
$3$ indices at least $\frac k2$ and at most $4\lceil \frac{2k}K\rceil$ indices less than $\frac k2$. Hence
$q\geq k-(\frac{2k}{K}+7)$. 

The best known polynomials for bounding border rank so far are the Koszul flattenings which
we review in \S\ref{kossect}.
We use them to prove the following lemmas:

\begin{lemma}\label{claim:no5large}
If  
\be\label{claimhyp}
|I_{+}\cap[0,\frac{3}{7}k]|> \frac{4}{K}\frac{3}{7}k+4\ \ {\rm or}\ \ |I_{-}\cap[-\frac{3}{7}k,0]|> \frac{4}{K}\frac{3}{7}k+4,
\ene
then $\ur(T_k)\geq (2+\ep)(2k+1)$ for $k>\frac{189K+21K\ep-252+56K^2}{K+78-42K\ep}$.
\end{lemma}

\begin{proof}  
First suppose $|I_{+}\cap[0,\frac{3}{7}k]|> \frac{4}{K}\frac{3}{7}k+4$.
We may find five indices $i_1<\dots<i_5\leq \frac{3}{7}k$ in $I_{+}$ with $i_5-i_1\leq K$. Take
the  smallest possible such indices,
and  write $i_3-1=uK+r$, where $u=\lfloor \frac {i_3}k\rfloor$, then
dividing $\{ 1\hd i_3-1\}$ into $u+1$ sets, $\{ aK+1\hd (a+1)K\}$, $0\leq a\leq u-1$, and $\{ uK+1\hd i_3-1\}$,
the first $u$ cannot contain more than $4$ surviving indices, and the last cannot contain more than $3$, so
this interval contains at most $4u+3$ surviving indices.
We conclude
  $q\geq [k-(\frac{2k}{K}+7)]+ i_3-\frac 4K (i_3-1)-3$.
Restrict $T$ to the subspace of $B^*\ot C^*$ where  $i_3$ becomes the diagonal.
By Lemma \ref{neardiagcase}, with $n=2k+1-i_3$ and $C=K$,  
we have $\ur(T_{k,q})\geq \frac 53(2k+1-i_3)-\frac 83K$.
Since  $i_3\leq \frac{3}{7}k-2$ we have
$$\ur(T_k)\geq  \ur(T_{k,q})+q\geq \frac{  56K^2 - (85k - 147)K + 78k - 252}{21K}.
$$
In order to have $\ur(T_k)\geq  (2+\ep)(2k+1)$ we 
need
$$
k\geq  \frac{7(8K^2 + 3K\ep  + 27K -36)}{K-42 K\ep - 78}
$$
At this point  we may assume $|I_{+}\cap[0,\frac{3}{7}k]|\leq \frac{4}{K}\frac{3}{7}k+4$ 
and thus  $q\geq k-(\frac{2k}{K}+7)+\frac{3(K-4)}{7K}k-4$. Using this, we   conclude   in the
case $|I_{-}\cap[-\frac{3}{7}k,0]|> \frac{4}{K}\frac{3}{7}k+4$ by the same reasoning as above with an even better bound.  
\end{proof}

In light of Lemma \ref{claim:no5large}, we may henceforth assume
\be\label{qbnd}
q\geq  [k-(\frac{2k}{K}+7)]+[\frac 37k-(\frac{12k}{7K}+4)]=\frac{10}{7}k-\frac{26k}{7K}-11.
\ene
 
\begin{lemma} \label{lem:far} Assumptions as above.
If  $|I_{-}\cap ( \frac k2,k]|\geq 3$,
then
$\ur(T_k)\geq (2+\ep)(2k+1)$ when $ 
K\geq 
\frac{39}{1-21\ep}
$ 
and
$ 
k\geq \frac{21(13+\ep)K}{ 2(K-21\ep K -39)}
$.
\end{lemma}
\begin{proof}
Let $y_3,y_4,y_5$ be the three indices in $I_{-}$ with absolute value greater than $\frac{k}{2}$. Let 
$y_1$ and $y_2$ be respectively the largest and smallest index in $I_{+}\setminus [0, \frac{k}{2} ]$, so we may add
$\frac k2-(y_1-y_2)-1$ to \eqref{qbnd}. 
By Lemma \ref{primecase} with $n=2k+1$ we obtain:
\begin{align*}\ur(T_k)&\geq \ur(T_{k,q})+q\\
&\geq 
2k+1+y_1-\frac{2}{3}y_2+[\frac{10}{7}k-\frac{26k}{7K}-11]+[\frac{k}{2}-(y_1-y_2)-1]\\
&\geq \frac{86}{21}k-\frac{26}{7}\frac kK-7,
\end{align*}
where to obtain the last  line we used that $y_2> \frac{k}{2}$.


This bound will be greater than $(2+\ep)(2k+1)$ when
$$
K\geq 
\frac{39}{1-21\ep}
$$
and
$$
k\geq\frac{21(13+\ep)K}{ 2(K-21\ep K -39)}
$$
so  we conclude.
\end{proof}
\begin{lemma} \label{lem:near}
If  $I_{-}$ has five  indices with absolute value smaller than $\frac{k}{2}$
such that the difference of the smallest and largest is at most $K$, 
then
$\ur(T)\geq (2+\ep)(2k+1)$ when
$
K\geq \frac{78}{1-42\ep}
$
and
$
k\geq \frac{21(13+ \ep)K}{K-42\ep K-228}.
$
\end{lemma}
\begin{proof}
By Lemma \ref{claim:no5large} we may assume that these five indices have absolute value at least $\frac{3}{7}k-3$.
We may proceed as in Lemma \ref{lem:far} fixing $y_2$ and $y_1$ to be respectively the smallest and 
largest index in $I_{+}\setminus [0,\frac{4}{7}k]$, so we may add $\frac{3}7k-(y_1-y_2)-1$ to \eqref{qbnd}.

We obtain:
\begin{align*}
\ur(T_k)&\geq \ur(T_{k,q})+q\\
&\geq \left(2k+1+y_1-\frac 23 y_2\right)+ \frac{10}{7}k-\frac{26k}{7K}-11+\frac 37 k-(y_1-y_2)-1\\
&\geq (4+\frac 1{21})k-\frac{26}{7}\frac kK -11,
\end{align*}
where in the last line we used that $y_2\geq \frac 47k$.
This is at least $(2+\ep)m$ when
\be\label{Keqn}
K\geq \frac{78}{1-42\ep}
\ene
and
$$
k\geq \frac{21(13+ \ep)K}{K-42\ep K-78}
$$
and we conclude.
\end{proof}

\begin{proof}[Proof of Theorem \ref{mainthm}]
One of the cases in Lemmas \ref{claim:no5large},\ref{lem:far}, or  \ref{lem:near} must hold, which proves the theorem for $k$ large enough.

In order for assumptions of all lemmas to hold,  we need $K\geq \frac{78}{1-42\ep}$ and 
$$
k\geq\tmax\left\{ \frac{21(13+ \ep)K}{K-42\ep K-78}, \frac{189K+21K\ep-252+56K^2}{K+78-42K\ep}\right\}.
$$
Taking the estimate $K=\frac{78}{1-42\ep}+1$
we obtain
$$
k\geq 
\frac{ 37044\ep^3 - 82908\ep^2 - 983829\ep + 364175}{(1-42\ep)^3}
$$
In particular, when 
$\ep=.001$ it suffices to take  $k\geq 413085$, when 
$\ep=.01$,  it suffices to take  $k\geq 1.82\times 10^6$,  
 and when $\ep=.02$  it suffices to take 
$k\geq 8.41\times 10^7$.
\end{proof}

\section{Koszul  flattenings}\label{kossect}

In \S\ref {subsec: generalized flattenings} we review Koszul flattenings.
We then, in \S\ref{explicitpf}, prove Theorem \ref{thm:explicit}. In \S\ref{caseone}, \S\ref{casetwo},
we   prove lower border rank bounds for   two types of tensors arising in the proof of Theorem \ref{mainthm}.

\subsection{Definition}\label{subsec: generalized flattenings}
Fix $p<\frac 12 \tdim A$. There is a natural inclusion
$A\subset \La p A^*\ot \La {p+1}A$ inducing an
inclusion $A\ot B\ot C\subset (\La p A^*\ot B)\ot (\La{p+1}A\ot C)$.
Thus given $T\in A\ot B\ot C$, we obtain a linear map
$T_{A}^{\ww p}: \La p A\ot B^*\ra \La{p+1}A\ot C$.

In  bases $\{a_i\}$, $\{b_j\}$, $\{ c_k\}$ of   $A,B,C$, where 
$T = \sum_{ijk} T^{ijk} a_i \otimes b_j \otimes c_k \in A \otimes B \otimes C$, 
  the linear map $T_A^{\ww p}$ is
\begin{align*}
T_A^{\ww p}: \Lambda^p A \ot B^* &\to \Lambda^{p+1}A\ot C \\
X\ot \beta & \mapsto \textsum_{ijk}T^{ijk}\beta(b_j)(a_i\wedge X) \ot c_k.
\end{align*}

Then  \cite[Proposition 4.1.1]{MR3081636} states 
\begin{equation}\label{kozinq}
\ur(T)\geq \frac{\trank(T_{A}^{\ww p})}{\binom {\tdim(A)-1}p}.
\end{equation}

In practice, one takes a subspace ${A'}^*\subseteq A^*$ of dimension $2p+1$ and 
restricts $T$ (considered as a trilinear form) to ${A'}^* \times  B^* \times 
C^*$ to get an optimal bound, so the denominator $\binom{\dim(A)-1}{p}$ is 
replaced by $\binom{2p}{p}$ in \eqref{kozinq}. Write $\phi : A \to A/ 
({A'}^*)^\perp=:A'$ for the projection onto the quotient: the corresponding Koszul 
flattening map gives a lower bound for $\uR(\phi(T))$, which, by linearity, is 
a lower bound for $\uR(T)$.

The $p=1$ case is equivalent to Strassen's equations \cite{Strassen505}, which, for tensors in 
$\BC^m\ot \BC^m\ot \BC^m$ can at best prove border rank lower bounds of $\frac 32m$.
We will utilize the $p=2$ case, which 
can at best prove border rank lower bounds of $\frac 53m$.
Thus we will need to use the border substitution method to kill off  at least a 
$(\frac 13+\ep)m$-dimensional space before using Koszul flattenings.
Notice that in the Set Up,   we actually kill off at a space of dimension  
$\sim\frac{1}{2}m$.

Let $T\in \BC^5\ot \BC^n\ot \BC^n$ and write
$T=a_1\ot X_1+\cdots + a_5\ot X_5$.
Order the bases of $\La 2\BC^5$, $a_i\ww a_j$, with  
$(ij)=12,13,14,15,23,24,25,34,35,45$
and of $\La 3\BC^5$ as
$234,235,245,345,123,124,125,134,135,145$.
Then the $10n\times 10n$ Koszul flattening matrix takes
the block form
\be\label{big10}T_{A}^{\ww 2}=
\begin{pmatrix}
0&0&0&0& X_4&-X_3&0&X_2&0&0 \\
0&0&0&0& X_5&0&-X_3&0&X_2&0 \\
0&0&0&0& 0&X_5&-X_4&0&0&X_2 \\
0&0&0&0& 0&0&0&X_5&-X_4&X_3 
\\
X_3&-X_2 &0& 0& X_1&0&0&0&0&0 \\
X_4&0 &-X_2& 0& 0&X_1&0&0&0&0 \\
X_5&0 &0&-X_2&        0&0&X_1&0&0&0 \\
0&X_4 &-X_3&0&  0&0&0&X_1&0&0 \\
0&X_5 &0&-X_3&        0&0&0&0&X_1&0 \\
0&0 &X_5&-X_4&        0&0&0&0&0&X_1 \end{pmatrix}
\ene

In our case all the matrices $X_i$ will be zero except on a semi-diagonal.

If $X_1$ is invertible, then, the rank of the matrix \eqref{big10} equals
$6n$ plus the rank of 
the block $4n\times 4n$ matrix with block entries
$( X_iX_1^{-1} X_j-X_jX_1^{-1}X_i)$.
(One gets this matrix with blocks permuted using the basis choice above.)
To see this use the basic identity, for a block matrix with 
$Z$ square and invertible: 
$$
\begin{pmatrix}
0&X\\ Y&Z\end{pmatrix}
=
\begin{pmatrix} 
\Id & XZ\inv \\ 0 &\Id\end{pmatrix}
\begin{pmatrix}
-XZ\inv Y & 0 \\ 0&Z\end{pmatrix} 
\begin{pmatrix} \Id & 0\\
Z\inv Y & \Id
\end{pmatrix}.
$$

We record the above observation:

\begin{proposition}\label{blockred} Let $T=a_1\ot X_1+\cdots + a_5\ot X_5\in \BC^5\ot \BC^n\ot \BC^n$ and let $X_1$ be
of   rank $n$.
If   the rank of the $4n\times 4n$ block matrix with size $n\times n$  blocks  $   X_iX_1^{-1} X_j-X_jX_1^{-1}X_i $
is at least 
$R$, then $\trank(T_{A}^{\ww 2})\geq R+6n$.
\end{proposition}

\subsection{Proof of Theorem \ref{thm:explicit}}\label{explicitpf}
We first project $T_6$ to to $\CC^9\otimes \CC^{13}\otimes \CC^{13}$ by considering all possible coordinate projections. By Proposition \ref{tightbapolarok} for one of the projections the border rank drops at least by four. In the algorithm in \S\ref{codesect} all possible ${13}\choose {4}$ projections are considered. 

Next, in each case we consider a general projection to $\CC^5\otimes \CC^{13}\otimes \CC^{13}$. Finally we apply $p=2$ Koszul flattenings described in Section \ref{subsec: generalized flattenings}. The algorithm checks the rank of the flattening matrix is always at least $127$.
We provide code to carry this out in \S\ref{codesect}. Thus:
$$\ur (T_6)\geq \lceil \frac{127}{6}\rceil+4=22+4=26.$$
\qed

\subsection{Border rank lower bounds for tensors  with $T(A^*)$ supported on the diagonal and four  semi-diagonals close to the diagonal}\label{caseone}

\begin{lemma}\label{neardiagcase}
Let $T=a_1\ot X_1+\cdots + a_5\ot X_5\in \BC^5\ot \BC^n\ot \BC^n$ be a tensor such that
$X_1$ is diagonal with nonzero entries on the diagonal, 
 $X_2,X_3$ are semi-diagonals supported on   semi-diagonals
$x_2< x_3$
above the main diagonal and $X_4,X_5$ are semi-diagonals supported on semi-diagonals $x_4<x_5$
below the main diagonal, and assume $x_s<C$ for $s=2\hd 5$ for some constant $C$.

If the entries of $T$ are generic,  then the $p=2$ Koszul flattening matrix drops rank by at most
$16C$.
In particular, 
$\ur(T)\geq \frac 53 n-\frac 83 C$. 
 
Further, if $T$ is the restriction of a subtensor of the tensor $T_k$, where $X_1$ corresponds to a semidiagonal of $T$ of distance at most $\frac{3}{7}k$ to the main diagonal, the same bound holds.
\end{lemma}

\begin{proof}

As the flattening of $T$ contains an invertible matrix, the rank of the second Koszul flattening drops exactly be the drop of the rank of the $4n\times 4n$ commutator matrix described in Proposition \ref{blockred}. This matrix is depicted in Figure \ref{fig:4x4}.

Write $M_{ts}:=X_tX_1^{-1} X_s-X_sX_1^{-1}X_t$, $t<s$, for the four distinct  size $n\times n$ block matrices appearing in Proposition 
\ref{blockred} (here $M_{st}=-M_{ts}$). 
The $M_{ts}$   are semi-diagonal
matrices. Explicitly, $M_{23}$ is supported on semi-diagonal $x_2+x_3$,
$M_{24}$ is supported on semi-diagonal $x_2-x_4$, $M_{25}$ is supported on semi-diagonal $x_2-x_5$,
$M_{34}$ is supported on semi-diagonal $x_3-x_4$, $M_{35}$ is supported on semi-diagonal $x_3-x_5$,
$M_{45}$ is supported on semi-diagonal $-(x_4+x_5)$,
where negative indices mean the semi-diagonal is below the main diagonal. We note that some of the entries on these semi-diagonals may be equal to zero.

By permuting the basis vectors, one obtains a  block diagonal matrix, where the  diagonal blocks are
of size $4\times 4$, namely we group columns
$x_2+j$, $x_3+j$, $-x_4+j$, $-x_5+j$ 
and rows $j-x_2,j-x_3,j+x_4,j+x_5$, and there are $n$ such blocks, as in Figure \ref{fig:4x4}.
Call the resulting $4\times 4$  blocks $N^j$, $1\leq j\leq n$.

It remains to show   almost all matrices $N^j$ have full rank. We note that each of the matrices $M_{ts}$ has at most $4C$ entries on the distinguished semi-diagonal that are not a binomial in the entries of $X_i$'s. These non-binomial entries appear only if one of the column or row indices of $M_{ts}$ is either smaller than $2C$ or greater than $n-2C$. We focus on all $N^j$'s contained in the other rows and columns, i.e.~$2C<j<n-2C$. A direct computation shows that the determinant of each such $N^j$ is, after clearing denominators, an explicit nonzero polynomial $P^j$ in the entries of the matrices $X_i$. 

Hence, whenever the evaluations of the polynomials $P^j$ are nonzero the bounds on ranks given in Lemma
\ref{neardiagcase} hold. This holds if the entries of the tensor are generic. 

We now prove the last statement of the lemma, i.e.~that the rank condition holds when $T$ is a restriction of a subtensor of $T_k$.
For this we note that each $P^j$, irrespective of the choice of $T$ and $j$, depends only on variables corresponding to entries that are contained in a $4C\times 4C\times 4C$ sub-tensor of $T_k$  made of consecutive indices. As $C$ is a fixed constant and $P^j$ are explicit polynomials, the non-vanishing can be achieved by filling one $4C\times 4C\times 4C$ cube with numbers such that any subset of them is not a root of any of the $P^j$'s and then assigning the values of the tensor periodically as in Lemma \ref{explb}.
\end{proof}
\begin{figure}[!htb]\begin{center}
\includegraphics[scale=0.5]{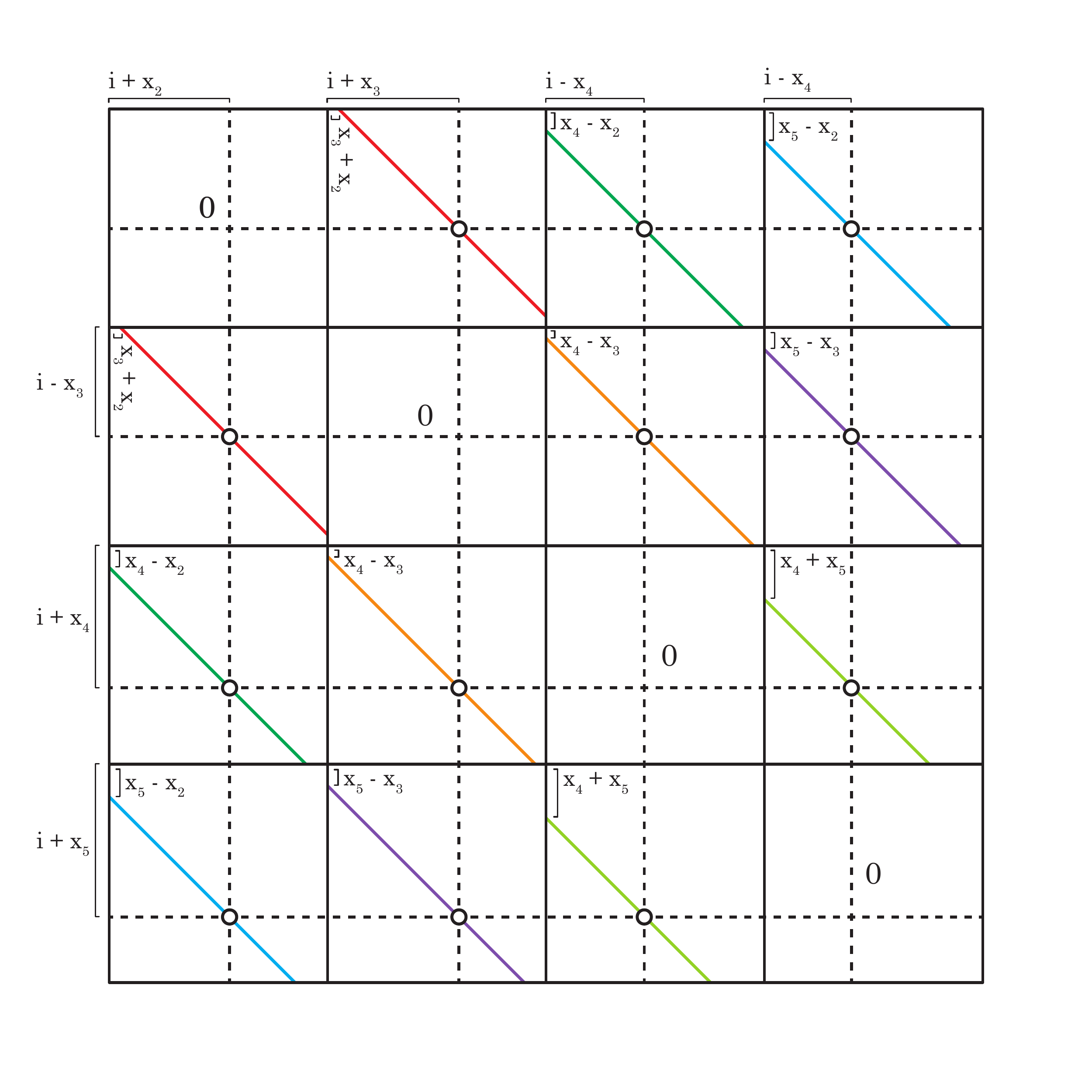}
\caption{The $i$-th $4\times 4$ block $N^i$ has entries corresponding to the intersection of the dashed lines}\label{fig:4x4}
\end{center}
\end{figure}

\subsection{Border rank lower bounds for tensors  supported on $5$ semi-diagonals not equal to the diagonal}\label{casetwo}

\begin{lemma}\label{primecase}
Let $T\in \BC^5\ot \BC^{n}\ot \BC^{n}$ be a tensor supported on 
three semi-diagonals at distances $y_3<y_4<y_5<\frac{n}{2}$ below the diagonal and two semi-diagonals
at distances $y_1>y_2$ above the diagonal such that $y_1-y_2<y_3$.
Assume the entries of $T$ are distinct primes.
Then
$$
\ur(T)\geq n+y_1-\frac 23 y_2.
$$
\end{lemma}

\begin{proof}
Consider the second Koszul flattening matrix of $T$ given in \eqref{big10}. The bound of $\ur(T)$ will follow by estimating the rank of this matrix.
 
Consider,  in each of the last six blocks,
the last $n-y_1$ columns. We obtain $6n-6y_1$ independent columns  due to the matrix $X_1$ in  \eqref{big10}. 

We now restrict our attention to columns that are still in the last six blocks, but
just consider the  first $y_1$ columns in each block. We focus on the upper four horizontal blocks, obtaining a $4n\times 6y_1$ matrix, as in Figure \ref{fig:4x6}. We claim that the rank of this matrix is at least $4y_1+2(y_1-y_2)$.

In this matrix the second, third and sixth vertical blocks give independent columns, due to the matrices that are $y_3$ below the diagonal (blue in Figure \ref{fig:4x6}). We claim that the $y_1$ columns in the first vertical block are also independent of  them. For contradiction suppose there is a linear relation involving $i$-th column in the first block. The $(i+y_4)$-th entry of the first block of this column may only be canceled using the $(i+y_4-y_3)$-th column in the second block. Then the contribution in the third block may only be canceled by the $(i+y_5-y_3)$-th column in the third vertical block. Now the entries in the second horizontal block cancel if and only if a binomial cubic equation on the entries is satisfied. This is not possible, as the entries are distinct primes. 

So far we obtain the lower bound on the rank of matrix in Figure \ref{fig:4x6} equal to $4y_1$. By adding the last $(y_1-y_2)$ columns of fourth and fifth block we may find $2(y_1-y_2)$ more independent columns.

We now focus on the first four vertical blocks in \eqref{big10}. Among those,
the first four horizontal blocks are zero, thus we restrict to the  last six. We further restrict to the rows that are zero in the last six vertical blocks. We obtain a $6\times 4$ block matrix. We lower bound   the rank of this matrix in 
the same way we attained the bound for the matrix in Figure \ref{fig:4x6}.

In total we obtain that the rank of \eqref{big10} is at least
$$6(n-y_1)+2(4y_1+2(y_1-y_2))=6n+2(3y_1-2y_2).$$
This exactly translates to the bound on the border rank as in the statement of the lemma.
\end{proof}

\begin{figure}[!htb]\begin{center}
\includegraphics[scale=0.5]{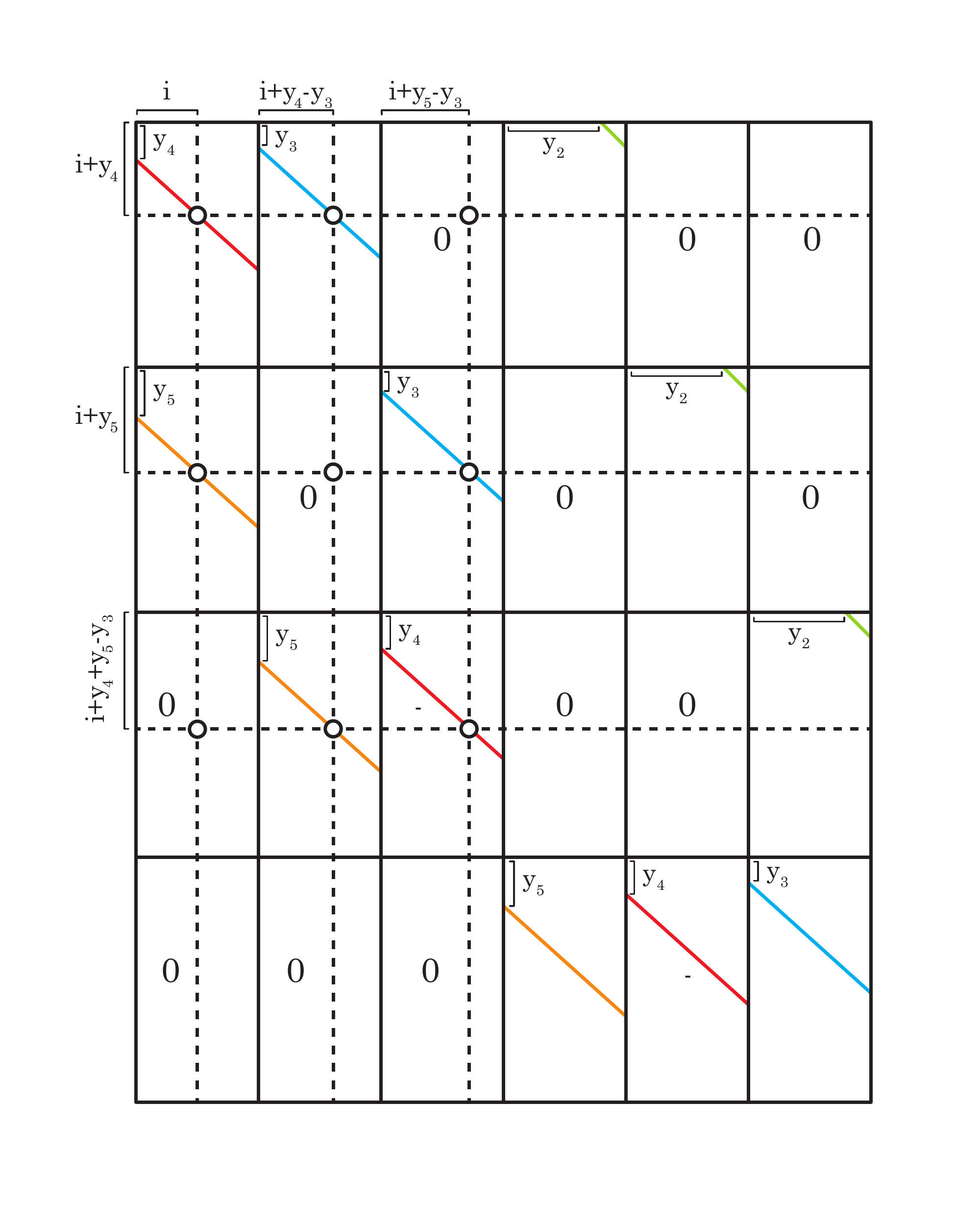}
\caption{The size $4n\times 6y_1$ submatrix of the  $4n\times 6n$ matrix.}
\label{fig:4x6}
\end{center}
\end{figure}

\newpage

\section{Code for the proof of Theorem \ref{thm:explicit}}\label{codesect}
\begin{verbatim}
S=QQ[w_0..w_4];
rs=res ideal basis(1,S);
kosz=rs.dd_3;
R=QQ[a_(-6)..a_6,b_(-6)..b_6,c_(-6)..c_6,w_0..w_4] 
T6:=0
for i from -6 to 6 do
(
 for j from max(-6,-i-6) to min(6,-i+6) do
  ( 
   T6=T6+(2^i+3^j-7*i*j)*a_i*b_j*c_(-i-j)
  ); 
)
by=matrix{{b_(-6)..b_6}}; cz=matrix{{c_(-6)..c_6}};
L=subsets(-6..6,4); zle={};
for m in L do (
subp={};
 for mm from -6 to 6 do(
  if (not member(mm,m)) then(
   subp=subp|{a_mm=>(sum for i from 0 to 4 list (random(0,R)*w_i))};
  );
 );
 l= rank diff(cz,diff(transpose by,diff(sub(kosz,R),sub(T6,subp))));
if (l<127) then zle=zle|{m};
)
#zle
\end{verbatim}

\bibliographystyle{amsplain}
 
\bibliography{Lmatrix}

\def\cdprime{$''$} \def\cprime{$'$} \def\cprime{$'$} \def\cprime{$'$}
  \def\Dbar{\leavevmode\lower.6ex\hbox to 0pt{\hskip-.23ex \accent"16\hss}D}
  \def\cprime{$'$} \def\cprime{$'$} \def\cdprime{$''$} \def\cprime{$'$}
  \def\cprime{$'$} \def\Dbar{\leavevmode\lower.6ex\hbox to 0pt{\hskip-.23ex
  \accent"16\hss}D} \def\cprime{$'$} \def\cprime{$'$} \def\cprime{$'$}
  \def\cprime{$'$} \def\Dbar{\leavevmode\lower.6ex\hbox to 0pt{\hskip-.23ex
  \accent"16\hss}D} \def\cprime{$'$} \def\cprime{$'$}
\providecommand{\bysame}{\leavevmode\hbox to3em{\hrulefill}\thinspace}
\providecommand{\MR}{\relax\ifhmode\unskip\space\fi MR }
\providecommand{\MRhref}[2]{%
  \href{http://www.ams.org/mathscinet-getitem?mr=#1}{#2}
}
\providecommand{\href}[2]{#2}
\begin{thebibliography}{10}

\bibitem{MR3025382}
Boris Alexeev, Michael~A. Forbes, and Jacob Tsimerman, \emph{Tensor rank: some
  lower and upper bounds}, 26th {A}nnual {IEEE} {C}onference on {C}omputational
  {C}omplexity, IEEE Computer Soc., Los Alamitos, CA, 2011, pp.~283--291.
  \MR{3025382}

\bibitem{CGLVW}
A.~{Conner}, F.~{Gesmundo}, J.~M. {Landsberg}, E.~{Ventura}, and Y.~{Wang},
  \emph{{A geometric study of {S}trassen's asymptotic rank conjecture and its
  variants}}, ArXiv e-prints arXiv:1811.05511 (2018).

\bibitem{MR3320240}
J.~M. Landsberg, \emph{Nontriviality of equations and explicit tensors in
  {$\Bbb{C}^m\otimes\Bbb{C}^m\otimes\Bbb{C}^m$} of border rank at least
  {$2m-2$}}, J. Pure Appl. Algebra \textbf{219} (2015), no.~8, 3677--3684.
  \MR{3320240}

\bibitem{MR3633766}
J.~M. Landsberg and Mateusz Micha{\l}ek, \emph{On the geometry of border rank
  decompositions for matrix multiplication and other tensors with symmetry},
  SIAM J. Appl. Algebra Geom. \textbf{1} (2017), no.~1, 2--19. \MR{3633766}

\bibitem{MR3081636}
J.~M. Landsberg and Giorgio Ottaviani, \emph{Equations for secant varieties of
  {V}eronese and other varieties}, Ann. Mat. Pura Appl. (4) \textbf{192}
  (2013), no.~4, 569--606. \MR{3081636}

\bibitem{MR3842382}
Joseph~M. Landsberg and Mateusz Micha{\l}ek, \emph{A {$2n^2-\log_2(n)-1$} lower
  bound for the border rank of matrix multiplication}, Int. Math. Res. Not.
  IMRN (2018), no.~15, 4722--4733. \MR{3842382}

\bibitem{MR87f:15017}
Thomas Lickteig, \emph{Typical tensorial rank}, Linear Algebra Appl.
  \textbf{69} (1985), 95--120. \MR{87f:15017}

\bibitem{rosser1962approximate}
J~Barkley Rosser, Lowell Schoenfeld, et~al., \emph{Approximate formulas for
  some functions of prime numbers}, Illinois Journal of Mathematics \textbf{6}
  (1962), no.~1, 64--94.

\bibitem{Strassen505}
V.~Strassen, \emph{Rank and optimal computation of generic tensors}, Linear
  Algebra Appl. \textbf{52/53} (1983), 645--685. \MR{85b:15039}

\bibitem{MR1089800}
\bysame, \emph{Degeneration and complexity of bilinear maps: some asymptotic
  spectra}, J. Reine Angew. Math. \textbf{413} (1991), 127--180. \MR{1089800}

\end{thebibliography}

 \end{document}